%% file: cake-cutting.tex
\documentclass[10pt]{article}
\input{preamble.tex}

\title{A Compositional Game to Fairly Divide Homogeneous Cake}
\author{Abel Jansma$^{1, 2}$}
\date{
    \textit{\footnotesize  $^1$Max Planck Institute for Mathematics in the Sciences, Leipzig, Germany \\
     $^2$School of Informatics, University of Edinburgh, UK
    }\\[2ex]%
    \today}

\begin{document}
\maketitle
\abstract{
The central question in the game theory of cake-cutting is how to fairly distribute a finite resource among multiple players. Most research has focused on how to do this for a heterogeneous cake in a situation where the players do not have access to each other's valuation function, but I argue that even sharing homogeneous cake can have interesting mechanism design. Here, I introduce a new game, based on the compositional structure of iterated cake-cutting, that in the case of a homogeneous cake has a Nash equilibrium where each of $n$ players gets $1/n$ of the cake. Furthermore, the equilibrium distribution is the result of just $n-1$ cuts, so each player gets a contiguous piece of cake. Naive composition of the `I cut you choose' rule leads to an exponentially unfair cake distribution with a Gini-coefficient that approaches 1, and suffers from a high Price of Anarchy. This cost is completely eliminated by the proposed \textit{Biggest Player} rule for composition which achieves decentralised and asynchronous fairness at linear Robertson-Webb complexity. After introducing the game, proving the fairness of the equilibrium, and analysing the incentive structure, the game is implemented in Haskell and the Open Game engine to make the compositional structure explicit.
}

\begin{multicols}{2}
\setlength{\parindent}{0pt}
\setlength{\parskip}{0.7em}

\section{Introduction}
\subsection{Not a piece of cake} \label{sec:pieceOfCake}
Fairly distributing a finite resource among $n$ players is a challenging but relevant and ubiquitous problem. It is challenging because---in the absence of a dictatorship---the players need to agree to the mechanism that distributes the resource, and will only do so if they believe they are not at a disadvantage relative to other players. Matters are complicated further when the resource to be divided is heterogeneous and the players have different valuation functions, \ie not everybody values each part equally. In that case, what is meant with a \textit{fair} distribution is not immediately clear, and can be defined in different ways. Solutions to this problem have been extensively studied \cite{robertson1998cake}, both in the case where the shared resource is positive, in which case it is usually abstracted as a cake \cite{steinhaus1948}, and in the case where the resource is negative, in which case it is usually abstracted as a chore \cite{martin1978aha}. Focusing on the positive case, distributing the resource can thus be abstracted to the problem of fair cake-cutting. 

The most famous solution to the simple case of 2 players is called `I cut you choose', and has seen ubiquitous use. It was used by Prometheus and Zeus to divide an ox \cite{hesiodTheogony}, by the United Nations to divide oceans \cite{unclos1982}, and by children everywhere to divide actual cake. The game is structured as follows: the \textit{cutter}, $P_1$, divides the cake into two pieces, and the \textit{chooser}, $P_2$, gets to pick one of the two pieces, leaving the cutter with the leftover piece. This method is fair in the sense that the cut that maximises the cutter's payoff is the one that divides the cake into two pieces of equal value. However, if the cake is heterogeneous, then the players might disagree on the value of the different pieces. In that case, while the cutter will consider their division to be perfectly fair, the chooser might value one of the two pieces more than the other, and thus end up with a piece that contains more than half of all the value, according to the chooser's valuation. The complexity of fair cake-cutting, as well as the wide range of important applications, has led to various definitions of \textit{fairness}.

\subsection{What is fair?}

Let $P_i$, where $i \in \mathbb{N}_{\leq n} = \{1, \ldots, n\}$, be the $i$th player in an $n$-player cake-cutting game, where the cake can be represented as a set $C$. Let $X_i \subseteq C$ be the piece of cake that $P_i$ gets in a given partition $X: \mathbb{N}_{\leq n} \to \mathcal{P}(C)$, such that $\cup_{i=1}^n X_i = C$. Let $V_i: \mathcal{P}(C) \to [0, 1]$ be the valuation function of $P_i$. Three important and different ways of defining fairness are then:
\begin{itemize}
    \item \textbf{Proportionality:} $X$ is a proportional partition if $\forall i: V_i(X_i)\geq \frac{V_i(C)}{n}$. That is, if each player believes they ended up with \textit{at least} $\frac{1}{n}$ of the total cake. 
    \item \textbf{Envy-freeness:} $X$ is called envy-free if no player would want to switch their piece for someone else's, \ie if $\forall i\ \forall j: V_i(X_i)\geq V_i(X_j)$. 
    \item \textbf{Equitability:} $X$ is equitable if each player values their piece equally, \ie if $\forall i\ \forall j: V_i(X_i) = V_j (X_j)$. 
\end{itemize}

The canonical example of a game with an equilibrium that is both proportional and envy-free is the `I cut you choose' game. However, as discussed in Section \ref{sec:pieceOfCake}, there is a type of meta-envy in the sense that equitability is not guaranteed, so it is sometimes more profitable to be the chooser than the cutter. A distribution mechanism where there is no envy about role assignments is called \textit{meta-envy-free}, or \textit{symmetric fair} \cite{manabe2010}. Finally, even when a particular partition $X$ is fair in any of the senses above, it might not be \textit{optimal}---a different partition $Y$ might leave the players happier. It is said that $Y$ \textit{Pareto-dominates} $X$ when at least one player is happier with $Y$ than with $X$, and all other players are at least as happy with $Y$ as $X$, \ie if $\exists i: V_i(Y_i)>V_i(X_i)$ and $\forall i\ V_i(Y_i) \geq V_i(X_i)$. $X$ is said to be \textit{Pareto-optimal} (or \textit{Pareto-efficient}) if there is no such $Y$.

\subsection{More than 2 players}
Designing fair cake-cutting rules for more than two players is the focus of most cake-cutting research, and has led to numerous new division games. Famously, the so-called \textit{moving-knife} procedures \cite{dubins1961, stromquist1980, austin1982, peterson2002} give an envy-free solution to the 3-player game, and the \textit{last-diminisher} method \cite{steinhaus1948} offers a proportional solution to the $n$-player game, but all these approaches require continuous and universally agreed upon time, perfect and direct communication between all players, an external referee, or an unbounded number of cuts that leave each player with many disconnected pieces of cake. In fact, a finite algorithm (\ie involving at most a finite number of steps) with a proportional equilibrium distribution of contiguous pieces does not exist:

\begin{theorem}[Impossibility of fair contiguous pieces \cite{robertson1998cake}]
    No finite algorithm can guarantee each of $n$ players at least $\frac{1}{n}$ of the cake using only $n-1$ cuts when $n\geq 3$.
\end{theorem}
The same impossibility theorem has also been proven for envy-free distributions for $n\geq 3$ \cite{stromquist2008envy}.

However, if the cake is homogeneous, \ie if $\forall i\ V_i = V$, then proportionality, envy-freeness, equitability, and pareto-optimality all coincide, and can simply be summarised as \textit{fairness}. Dividing such a homogeneous cake is---in a sense---trivial, since there is an obvious partition that all players agree is fair (namely $\forall i: V(X_i)=\frac{1}{n}$), which the first player can immediately achieve with $n-1$ cuts. One could imagine a game that lets $P_1$ come up with the partition, but punishes $P_1$ by an arbitrary amount whenever they do not give each player $\frac{1}{n}$ of the cake. If the punishment is harsh enough, then this game has a (trivially) fair equilibrium. However, not only is this a very contrived game with uninteresting dynamics, it also requires identifying a player that can exchange information and cakes with all other players. One could get around this by passing the cake to the next player and only allowing each player to cut off a piece of size $\frac{1}{n}$, but this requires each player to know the original size of the cake. This information has to be passed on along with the cake, so requires that all players trust each other to do so honestly, even though there is a strong incentive to be dishonest. Furthermore, it is not meta-envy-free. If the partition is not absolutely perfect, then whenever $V(X_1)>\frac{1}{n}$, the other players cannot judge if this is the result of luck or malice, so they might still be envious of $P_1$. Such imperfect partitions naturally arise whenever $|C| \bmod n \neq 0$, or when perfect division is simply not possible (as is the case with real cake-cutting). 

To relax the restrictions on the communication channels, and to decentralise the power, one could define the game starting from its compositional structure. Two examples of this---one exponentially unfair and the other fair---are explored in the following sections. 

\section{Compositional cake-cutting}
The development of new mathematical \cite{ghani2018} and software \cite{hedges2022} tools based on category theory has led to significant progress in the field of compositional game theory, and in part inspired the development of the game introduced here (see the Appendix for an implementation of compositional cake cutting in the Open Game Engine DSL \cite{hedges2022}). In compositional game theory, complex games are built by connecting smaller games through interfaces and explicitly defining the flows of information and payoff. A compositional, and explicitly \textit{open} version of `I cut you choose' would be a 2-player game among $P_a$ and $P_b$ that has as an input a piece of cake $X_a$ that $P_a$ then cuts, after which $P_b$ chooses a piece. $P_a$ then gets the leftover piece, and $P_b$ leaves the game with their piece, but gets no payoff yet. This can be composed into an $n$-player game by inputting the full cake into a game with $P_1$ and $P_2$, and linking consecutive versions of this game together, each time propagating the chosen piece as the input to the next game. Denoting the 2-player game among $P_a$ and $P_b$ with piece $X_i$ as input by $\mathcal{G}(X_i, P_a, P_b)$, the 4-player composition looks like this:
\[\begin{tikzcd}
	{\mathcal{G}(C, P_1, P_2)} & {\mathcal{G}(X_2, P_2, P_3)} & {\mathcal{G}(X_3, P_3, P_4)}
	\arrow["{X_2}", from=1-1, to=1-2]
	\arrow["{X_3}", from=1-2, to=1-3]
\end{tikzcd}\]
To make this a `closed' game, the last player ($P_4$) should still get a payoff, which can just be set to whatever piece they are left with. This compositional game puts less restrictions on the communication channels among the players since it is only required that pairs can communicate. It is the simplest possible compositional rule for iterated `I cut you choose', so will be referred to as \textit{vanilla} composition. Note, however, that this game's equilibrium is far from fair. When the cake is homogeneous, each cutter will just cut whatever piece they have exactly in half. This means that the $m$th player will end up with $2^{-m}$ of the cake, except for the $n$th player who ends up with $2^{-(n-1)}$ of the cake. This cake distribution is \textit{exponentially unfair} (at least asymptotically so in $n$). The fairness of a distribution can be measured with the Gini coefficient, which is the relative mean absolute difference in wealth (\ie the amount of cake) in a population:
\begin{align}
    G = \frac{\sum_{i=1}^n \sum_{j=1}^n \mid V(X_i) - V(X_j)\mid}{2 n \sum_{k=1}^n V(X_k)}
\end{align}
When all players have non-negative valuation functions, $G$ is 0 for perfect equality, and 1 for perfect inequality (\ie a single player has all value). The Gini-coefficient of iterated `I cut you choose' under vanilla composition is shown in Figure \ref{fig:Ginicoeff}, where it can be seen that $G$ approaches its maximum value of 1 as the number of players increases. In fact, it does so very quickly, and the 6-player game already has $G \approx 0.5$, which makes it about as unfair as the worst national income inequalities in the world \cite{gini_worldBank}.

As the number of players increases, the distribution approaches $V(X_i) = |C| 2^{-i}$. One can derive an analytic expression for the Gini-coefficient of this asymptotic distribution by writing $G$ using the alternative definition of the mean difference from \cite{glasser1962variance}:
\begin{align}
    G = \frac{2 \sum_{i=1}^n i V(X_i)}{n \sum_{i=1}^n V(X_i)} - \frac{n+1}{n}
\end{align}
where now the players $P_i$ are ordered such that $V(X_1)\leq (X_2) \leq \ldots \leq (X_n)$. Since the ordering is opposite in $V(X_i) = |C| 2^{-i}$, we instead adjust the rank index to get:
\begin{align}
    G &= \frac{2 \sum_{i=1}^n (n-i) 2^{-i}}{n} - \frac{n+1}{n}\\
    \intertext{which simplifies to}
    &= \frac{2(n + 2^{1-n} - 2)}{n} - \frac{n+1}{n}\\
    &= 2^{1-n} + 1 - \frac{3}{n}
    \intertext{which, as $n\to \infty$, simply approaches}
    &= 1 - \frac{3}{n}
\end{align}
That this approximation works well even for small $n$ can be seen in Figure \ref{fig:Ginicoeff}.

The naive composition made iterated `I cut you choose' exponentially unfair, which inspired the main contribution of this manuscript: a different compositional rule, called the \textit{Biggest Player} rule, for iterated cake-cutting among $n$ players that leads to a fair equilibrium for any $n$, using just $n-1$ cuts. 

\begin{figure}[H]
    \centering
    \includegraphics[width=0.4\textwidth]{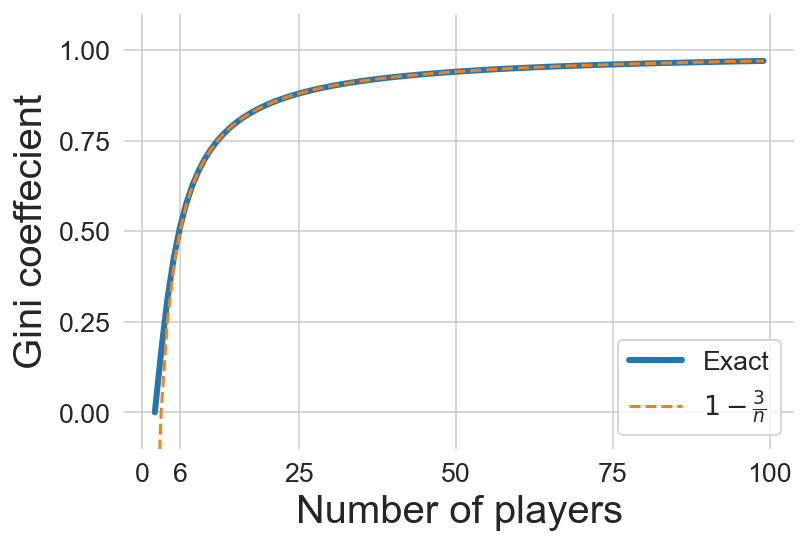}
    \caption{The Gini-coefficient of the exponentially unfair distribution of iterated `I cut you choose' as a function of the number of players.\label{fig:Ginicoeff}}
\end{figure}

\section{The \textit{Biggest Player} rule}
Consider the following game, which is just a compositional version of `I cut you choose', composed with the \textit{Biggest Player} rule:

\begin{definition}[The Biggest Player Rule for composition]
Let $\{P_1, P_2, \ldots, P_n\}$ be the players in the $n$-player game of dividing a homogeneous cake $C$. The \textit{Biggest Player} rule then prescribes the following:
    \begin{enumerate}
        \item $P_1$ cuts the cake in two, resulting in two pieces of sizes $(\alpha |C|, (1-\alpha)|C|)$, where $\alpha \in [0, 1]$. 
        \item $P_2$ chooses one of the two pieces. 
        \item If there are any players left who did not play yet: Let the last cutter and the chooser be $(P_\text{cut}, P_\text{choose})$, respectively, and the size of their pieces $(a, b)$, respectively. Then, let $P_\text{BP}=P_\text{cut}$ if $a\geq b$, and $P_\text{BP}=P_\text{choose}$ otherwise. $P_\text{BP}$ then has to cut their piece in two.
        \item A player that did not play yet chooses one of $P_\text{BP}$'s pieces. 
        \item Move to 3 if there are any players that have not played yet. 
    \end{enumerate}
\end{definition}

This is just iterated `I cut you choose', but after each round, the player who ends up with the biggest piece has to be the cutter in the next round. This can be summarised in the following diagram that shows that both a piece of cake and the identity of the next cutter are propagated to the next game:

\[\begin{tikzcd}
    {} & {\scriptstyle\mathcal{G}(C, P_1, P_2)} & {\scriptstyle{\mathcal{G}(X_{\text{BP}_1}, P_{\text{BP}_1}, P_3)}} & \ldots
    \arrow["{\scriptscriptstyle{P_1, C}}", from=1-1, to=1-2]
    \arrow["{\scriptscriptstyle{P_{\text{BP}_1}, X_{\text{BP}_1}}}", from=1-2, to=1-3]
    \arrow["{\scriptscriptstyle{P_{\text{BP}_2}, X_{\text{BP}_2}}}", from=1-3, to=1-4]
\end{tikzcd}\]

Sharing a homogeneous cake with the \textit{Biggest Player} rule has a fair equilibrium where each player ends up with a contiguous piece:
\begin{theorem}
    Cutting a homogeneous cake $C$ of size $|C|$ with the \textit{Biggest Player} rule has a Nash equilibrium at a distribution where each of $n$ players gets a contiguous piece of size $|C|/n$, involving $n-1$ cuts.
\end{theorem}
\begin{proof}
    The $n=1$ case is trivial. The $n=2$ case reduces to the game `I cut you choose', so inherits the proportional equilibrium. The proof then follows by induction: Assume that the $n$-player implementation has a proportional equilibrium, then consider the game with $n+1$ players. In this game, a proportional first cut would result in two pieces of respective sizes $(\frac{C}{n+1}, \frac{Cn}{n+1})$. Each of the two players then ends up with either a piece of size $\frac{C}{n+1}$, or a piece of size $\frac{Cn}{n+1}$ that then has to be shared with $n$ other players. Since the $n$-player game has a proportional equilibrium (by assumption), whoever ends up with the piece of size $\frac{Cn}{n+1}$, and all remaining players, will thus also get a proportional piece of size $\frac{C}{n+1}$. This shows that a proportional first cut results in a proportional distribution. Left to show is that there are no profitable deviations from this first proportional cut.

    If $P_1$ chooses to deviate by an $\epsilon$ such that $0<\epsilon<\frac{C(n-1)}{2(n+1)}$, then the resulting pieces are of size $(\frac{C}{n+1} + \epsilon, \frac{Cn}{n+1} - \epsilon)$, the first piece being the smallest. $P_\text{BP}$ ends up with a piece of size $\frac{Cn}{n+1} - \epsilon$, to be shared among $n$ players. The proportional equilibrium of the $n$-player game would then give $P_\text{BP}$ and each remaining player a piece of size $\frac{C}{n+1} - \epsilon/n$, which is smaller than true proportionality among $n+1$ players. Since $P_2$ profits from choosing the first piece they will do so, making this deviation not profitable for $P_1$. If $\epsilon > \frac{C(n-1)}{2(n+1)}$, then the first piece is the biggest, but this simply corresponds to choosing a new deviation $\delta = \frac{C(n-2)}{n+1} - \epsilon$ to which the same reasoning applies. If $\epsilon<0$, then $P_2$ can choose the bigger piece and end up with more than $\frac{C}{n+1}$ by playing the proportional $n$-player game with the remaining players, so each of these deviations is strictly loss-inducing for $P_1$.

    The only special case is a cut that leaves two pieces of equal size. In this case, the cutting player is always at a disadvantage, as they have to proportionally share their piece with all remaining players. As long as there are remaining players, this deviation from proportional cutting is thus strictly loss-inducing. 

    Therefore, the equilibrium of proportional distribution holds for $n+1$ players as long as it holds for $n$ players. Since it holds for $n=2$, it holds in general.
\end{proof}

The \textit{Biggest Player} rule therefore leads to a fair distribution of contiguous pieces (\ie involving the minimum number of cuts), that only requires communication in pairs, and makes malicious deviations strictly loss-inducing. This method is also of exceptionally low Robertson-Webb complexity. The Robertson-Webb complexity is the number cuts plus the number of evaluations that have to be made by the players \cite{robertson1998cake,woeginger2007complexity}. Vanilla composition requires the cutter to evaluate the incoming piece and cut it, and the chooser to evaluate the two pieces, resulting in $4+3(n-2)=\mathcal{O}_{RW}(n)$ queries (4 for the first 2-player game, 3 for every player added). The \textit{Biggest Player} rule does not require any additional cuts or evaluations, so is still of linear Robertson-Webb-complexity. For heterogeneous cake, envy-free cake cutting requires at least $\mathcal{O}(n^2)$ queries \cite{procaccia2009thou}, and proportional cutting requires $\mathcal{O}(n \log n)$ queries \cite{edmonds2006cake}.

Note that for $n>2$, the deviations are not \textit{symmetrically} loss-inducing. If the cutter cannot create a perfect proportional cut due to, say, finite precision or a cake for which $|C| \bmod n \neq 0$, then the cutter is incentivised to err on the side of making the smallest piece larger than $\frac{C}{n}$. The chooser will then certainly pick the smaller piece, but the loss induced by the deviation is then shared by all remaining players, rather than the cutter alone. This is shown in Figure \ref{fig:normPayoff}, where the difference in gradient of the payoff on the two sides of the equilibrium reflects the asymmetry. It can be clearly seen that the incentive asymmetry increases with the number of players. This illustrates that when perfectly proportional cuts are not possible, the \textit{Biggest Player} rule is envy-free, but not meta-envy-free: it is again better to be the chooser than the cutter. 

\begin{figure}[H]
    \centering
    \includegraphics[width=0.4\textwidth]{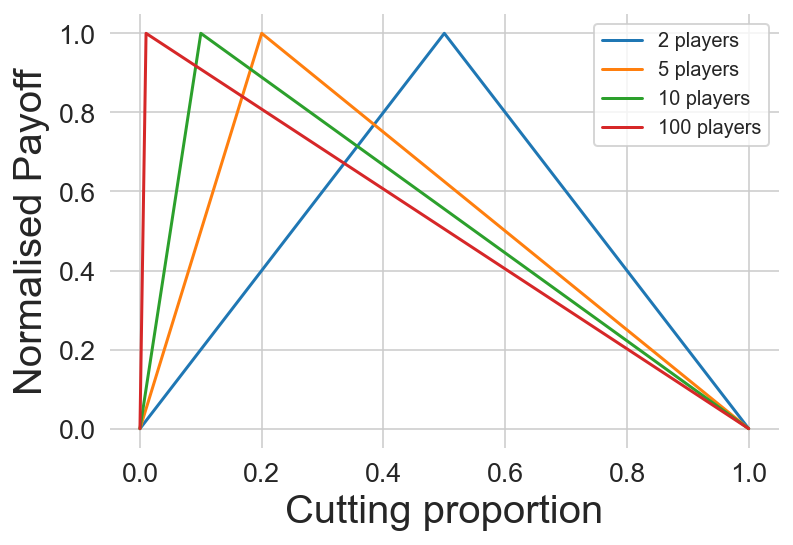}
    \caption{The (normalised) payoff is maximal for fair cuts, but the loss induced by deviations is asymmetric for $n>2$. \label{fig:normPayoff}}
\end{figure}

\section{The \textit{Biggest Player} rule eliminates the Price of Anarchy}
Vanilla composition of `I cut you choose' was shown to lead to an exponentially unfair equilibrium distribution. However, there is a trivial way to make the equilibrium fair, namely by letting a benevolent dictator impose a fair distribution. For a homogeneous cake, this is always possible and can lead to an optimally fair distribution, but comes at the price of centralised control. This unfairness is thus a reflection of the \textit{Price of Anarchy} \cite{koutsoupias1999}. This can be made precise as follows. Let the welfare $W(s)$ of a cake distribution $s$ be $W(s) = 1 - G(s)$, where $G(s)$ is the Gini coefficient of $s$. The Price of Anarchy in a given game $\G$ is defined as 
\begin{align}
    \text{PoA}^{\G} &= \frac{\text{max}_{s \in S} W(s)}{\text{max}_{s \in \tilde{S}} W(s)}\\
    &= \frac{\text{max}_{s \in S} 1 - G(s)}{\text{max}_{s \in \tilde{S}} 1 - G(s)}
\end{align}
where $S$ is the set of all cake distributions, and $\tilde{S}$ the set of equilibrium cake distributions. The Gini coefficient of vanilla composition was shown to approach 1, while a benevolent dictator could impose the perfectly fair distribution with Gini coefficient $0$. This means that the Price of Anarchy diverges for vanilla compositional cake-cutting, and in Figure \ref{fig:PoA} it can be seen that this divergence is almost perfectly linear. Using the asymptotic approximation $V(X_i)=|C|2^{-i}$, the Price of Anarchy can be seen to be $\text{PoA} = \frac{1}{1 - 1 - \frac{3}{n}} = \frac{n}{3}$, which again is a good approximation even for small $n$ (see Figure \ref{fig:PoA}). Using the \textit{Biggest Player} composition, the perfectly fair distribution is actually an equilibrium, so the Price of Anarchy is 1, which corresponds to decentralisation at no extra welfare cost. 
\begin{figure}[H]
    \centering
    \includegraphics[width=0.4\textwidth]{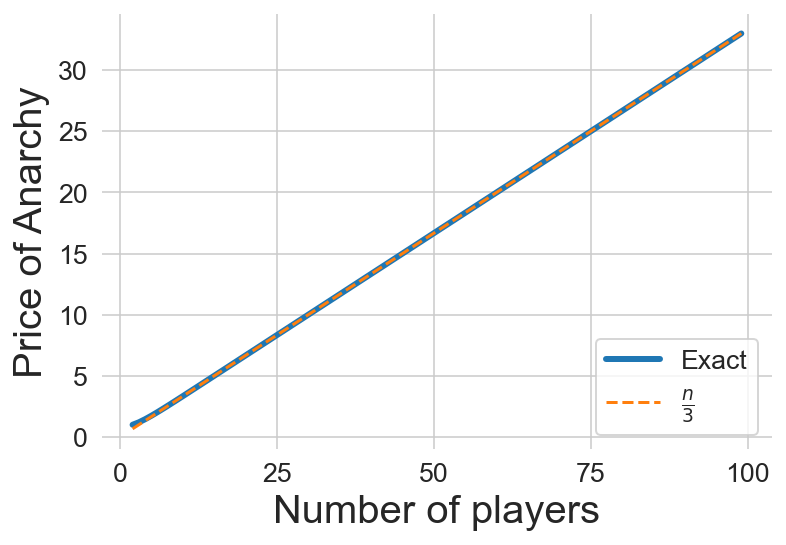}
    \caption{The Price of Anarchy grows linearly with the number of players. \label{fig:PoA}}
\end{figure}

\section{Discussion}
In this manuscript, I introduced a new compositional rule for homogeneous cake-cutting games, called the \textit{Biggest Player} rule. This rule implements iterative `I cut you choose', but lets the player with the biggest piece play again, which results in a fair equilibrium distribution. Since homogeneous cake-cutting implies that each player knows the valuation of all other players, it can also be solved by assigning a benevolent dictator who imposes a fair distribution. However, this comes at the cost of centralised control and communication, and for imperfect distributions makes it impossible to distinguish luck from malice. Both these problems are solved by the \textit{Biggest Player} rule, which eliminates the Price of Anarchy, decentralises power and communication, makes malice strictly loss-inducing, and has linear Robertson-Webb complexity. One might worry that identifying the \textit{Biggest Player} requires an external referee, but as long as all players share the same valuation function, each pair of players will agree on this issue, eliminating the need for a referee or dictator. Sharing a homogeneous resource with the \textit{Biggest Player} thus allows for fair but decentralised and asynchronous sharing, which could be beneficial in many real-life situations \cite{manabe2015cryptographic}. For example, it could be used to share a resource among a group of agents who can only communicate in pairs because they might not be in the same physical space, or on the same digital network. While the analysis here focused on a positive resource, the \textit{Biggest Player} rule can also be applied to negative resources, such as chores, where it could be used to fairly and asynchronously distribute the work among a group of people who can only communicate in pairs.

The aim of this paper has been twofold: to introduce the \textit{Biggest Player} rule, and to illustrate the power of compositional game theory as a framework for describing complex games. The \textit{Biggest Player} rule achieves fairness through the compositional structure of iterated games, not their individual implementation. Furthermore, the fact that the compositional structure was central to the definition made a software implementation natural, and analysis straightforward. As an example of this, the Appendix contains an implementation and discussion of the two iterated games in the Open Game Engine DSL \cite{hedges2022}.

\section*{Acknowledgements}
This work was supported in part by the Ethereum Foundation's Protocol Fellowship and an MRC Precision Medicine Grant (MR/N013166/1). The author thanks Rein Jansma and Barnabé Monnot for insightful discussions on game theory, and Jules Hedges, Philipp Zahn, and Fabrizio Genovese for helpful comments and suggestions regarding Open Games and the Open Game Engine. The author also thanks Gaia Joor for donating her birthday \textit{roti kukus} to science.

\bibliography{refs.bib}

\appendix

\section{Compositional cake-cutting as Open Games}
The Open Game engine \cite{hedges2022} was developed to analyse compositions of games, and is thus perfectly suited to analyse compositional cake-cutting. Here, I summarise the analysis of the two compositional cake-cutting games defined above: the exponentially unfair vanilla composition of `I cut you choose', and the \textit{Biggest Player} rule. The full code is available from \cite{Jansma2022}. 

In compositional game theory, games are modelled as a structure called \textit{lenses}. A lens $\mathcal{G}$ in a category $\mathcal{C}$ (with finite products) is an arrow between pairs of objects of $\mathcal{C}$, that is $\mathcal{G}: (X, S) \to (Y, R)$, composed of two morphisms: $f: X \to Y$ and $f^\#: X \times R \to S$. This can be represented as follows:

\begin{center} \begin{tikzpicture}
	\node (X) at (-2, .5) {$X$}; \node (Y) at (2, .5) {$Y$}; \node (R) at (2, -.5) {$R$}; \node (S) at (-2, -.5) {$S$};
	\node [rectangle, minimum height=1.5cm, minimum width=1cm, draw] (G) at (0, 0) {$\G$};
	\draw [->-] (X) to (G.west |- X); \draw [->-] (G.east |- Y) to (Y); \draw [->-] (R) to (G.east |- R); \draw [->-] (G.west |- S) to (S);
\end{tikzpicture} \end{center}

The horizontal direction can be interpreted as `time', so $f$ simply transforms $X$ into $Y$, but the morphism $f^\#$ takes the input $X$ as well as information $R$---sent back from the future---and uses these two to send information $S$ into the past. In practice, no information is being sent back and forth, and these directions rather reflect the reasoning of the agents in the game $\G$. For example, $R$ can correspond to the response of a player in some external game upon observing an input $Y$, about which the players in $\G$ can reason even before they output $Y$. 

This nicely captures the structure of compositional cake-cutting: at every round, a player is presented with two pieces of cake to choose from, and from this produces two new pieces of cake. They also inform the previous player of their choice, which they base on the cake they were presented with, as well as their reasoning about what next players are going to choose. This interpretation fills in the types on the dangling wires like this:
\begin{center} \begin{tikzpicture}
	\node (X) at (-2, .5) {$\mathbb{R}\times \mathbb{R}$}; \node (Y) at (2, .5) {$\mathbb{R}\times \mathbb{R}$}; \node (R) at (2, -.5) {$\mathbb{B}$}; \node (S) at (-2, -.5) {$\mathbb{B}$};
	\node [rectangle, minimum height=1.5cm, minimum width=1cm, draw] (G) at (0, 0) {$\PP_i$};
	\draw [->-] (X) to (G.west |- X); \draw [->-] (G.east |- Y) to (Y); \draw [->-] (R) to (G.east |- R); \draw [->-] (G.west |- S) to (S);
\end{tikzpicture} \end{center}

where $\mathbb{R} \times \mathbb{R}$ represents the sizes of the two pieces of cake being offered, and $\mathbb{B}$ represents the binary choice between the two pieces (say, 0 for the first piece, and 1 for the second). These games can then be composed sequentially, and put in a context by inputting values into all dangling wires:

\begin{center}\begin{tikzpicture}
    \node (X) at (-2, .5) {$(C, 0)$}; \node (Y) at (6, .5) {$(X_n, 0)$}; \node (R) at (6, -.5) {$\{1\}$}; \node (S) at (-2, -.5) {$\{0\}$};

    \node [rectangle, minimum height=1.5cm, minimum width=1cm, draw] (P1) at (0, 0) {$\PP_1$};
    \node [rectangle, minimum height=1.5cm, minimum width=1cm, draw] (P2) at (2, 0) {$\ldots$};
    \node [rectangle, minimum height=1.5cm, minimum width=1cm, draw] (Pn) at (4, 0) {$\PP_n$};

    \draw [->-] (X) to (P1.west |- X); 
    \draw [->-] (P1.west |- S) to (S); 

    \draw [->-] (P1.east |- Y) to node [above] {$\mathbb{R}^2$} (P2.west |- Y);
    \draw [->-] (P2.west |- R) to node [below] {$\mathbb{B}$} (P1.east |- R);
    \draw [->-] (P2.east |- Y) to node [above] {$\mathbb{R}^2$} (Pn.west |- Y);
    \draw [->-] (Pn.west |- R) to node [below] {$\mathbb{B}$} (P2.east |- R);

    \draw [->-] (Pn.east |- Y) to (Y); 
    \draw [->-] (R) to (Pn.east |- R); 
\end{tikzpicture}\end{center}

Note, however, that even though each game $\PP_i$ describes the behaviour of a single player, it is actually composed of two different actions: choosing and cutting. These can both be seen as separate games, linked only by the chosen piece, so that $\PP_i$ can be written as:
\begin{center}
    \begin{tikzpicture}
        \node (X) at (-2, .5) {$\mathbb{R}\times \mathbb{R}$}; \node (Y) at (4, .5) {$\mathbb{R}\times \mathbb{R}$}; \node (R) at (4, -.5) {$\mathbb{B}$}; \node (S) at (-2, -.5) {$\mathbb{B}$};
        \node (dY) at (4, 0.3){};
        \node [rectangle, minimum height=1.5cm, minimum width=1cm, draw] (G) at (0, 0) {$\PP_i^\text{choose}$};
        \node [rectangle, minimum height=1.5cm, minimum width=1cm, draw] (H) at (2, 0) {$\PP_i^\text{cut}$};
        \draw [->-] (X) to (G.west |- X); 
        \draw [->-] (G.east |- Y) to node [above] {$\mathbb{R}$} (H.west |- Y);
        \draw [->-] (H.east |- Y) to (Y);  
        \draw [->-] (R) to (H.east |- R); \draw [->-] (G.west |- R) to (S);
        \node [rectangle, minimum height=2.5cm, minimum width=3.6cm, draw, dashed, label=$\PP_i$] (BG) at (0.94, 0) {};
    \end{tikzpicture} 
\end{center}

Decomposing games to the fundamental actions like this reveals the compositional structure most directly, so this is the representation that the implementation in the Open Games DSL is based on.

\subsection{Vanilla compositional cake-cutting}
The interface between $\PP_\text{choose}$ and $\PP_\text{cut}$ is a single piece of cake, while from the outside $\PP_i$ looks like a game that maps cake offers into new offers, and propagates the binary choice responses along the backward wire. This is made explicit by the following implementation in the Haskell-based Open Games DSL:

\begin{lstlisting}[language=Haskell]
    openCakeCuttingUnit playerName = [opengame|

    inputs    : inputOffer  ;
    feedback  : playerResponse;
 
    :----------------------------:
    inputs    : inputOffer  ;
    feedback  : playerResponse;
    operation : CakeGame_choose playerName ;
    outputs   : chosenPiece ;
    returns   : ;
 
    inputs    : chosenPiece ;
    feedback  : ;
    operation : CakeGame_cut playerName ;
    outputs   : newOffer   ;
    returns   : newResponse;
 
    :----------------------------:
 
    outputs   : newOffer   ;
    returns   : newResponse;
    |]
\end{lstlisting}

The operations \texttt{CakeGame\_choose} and \texttt{CakeGame\_cut} are open games themselves, defined in a similar way (see \cite{Jansma2022} for the full implementation). This unit game can then be composed multiple times to very straighforwardly instantiate \eg a 3-player game as follows:

\begin{lstlisting}[language=Haskell]
openCakeSharing_threePlayers = [opengame|

   inputs    : inputOffer  ;
   feedback  : p1Response;

   :----------------------------:
   inputs    : inputOffer  ;
   feedback  : inputResponse;
   operation : openCakeCuttingUnit "p1" ;
   outputs   : newOffer1   ;
   returns   : newResponse1 ;

   inputs    : newOffer1  ;
   feedback  : newResponse1 ;
   operation : openCakeCuttingUnit "p2" ;
   outputs   : newOffer2   ;
   returns   : newResponse2 ;

   inputs    : newOffer2  ;
   feedback  : newResponse2 ;
   operation : openCakeCuttingUnit "p3" ;
   outputs   : newOffer3   ;
   returns   : newResponse3 ;
   
   :----------------------------:

   outputs   : newOffer3   ;
   returns   : newResponse3 ;
   |]
\end{lstlisting}

The equilibrium analysis, included in \cite{Jansma2022}, shows that the exponentially unfair distribution where player $i$ gets $2^{-i}$ of the cake (except for the last player who gets $2^{-n-1}$ of the cake) is indeed an equilibrium, whereas offering $\frac{1}{n}$ to the next player is not.

\subsection{Analysis of the \textit{Biggest Player} rule}
Similar to before, it makes sense to decompose the \textit{Biggest Player} compositional cake-cutting game into its fundamental games. However, while each player still has to choose a piece, not every player has to cut (players that choose the smaller of the two offered pieces do not cut). This means that the cutting game should take as an input the identity of the cutter (and, as before, be indexed by the identity of the chooser). This can be represented as follows:

\begin{center}
    \begin{tikzpicture}
        \node (X) at (-2, .2) {$\mathbb{R}$}; 
        \node (X1) at (-2, .7) {$\mathbb{N}^{\leq n}$}; 
        \node (X2) at (-2, .4) {}; 

        \node (Y) at (6, .2) {$\mathbb{R}$}; 
        \node (Y1) at (6, .7) {$\mathbb{N}^{\leq n}$};
        \node (Y2) at (6, .4) {}; 

        \node (R) at (6, -.6) {}; 
        \node (S) at (-2, -.6) {};

        \node [rectangle, minimum height=1.5cm, minimum width=1cm, draw] (G) at (0, 0) {$\PP_x^\text{cut}$};
        \node [rectangle, minimum height=0.8cm, minimum width=1cm, draw] (H) at (2, 0) {$\PP_i^\text{choose}$};
        \node [rectangle, minimum height=1.5cm, minimum width=1cm, draw] (u) at (4.2, 0) {$u(x, i)$};

        \draw [->-] (X) to (G.west |- X); 
        \draw [->-] (G.east |- Y) to node [above] {\small{$\mathbb{R}^2$}}  (H.west |- Y);
        
        \node (n) at (0.94, 0.2) {}; 
        \node [circle, scale=0.3, fill=black, draw] (o) at (1.1, 0.2) {};
        \draw [->-] (n) to [out=0, in=180] (1.5, 0.7) to (2.9, 0.7) to [out=0, in=180] (u.west |-X2) ;    
    
        \draw [->-] (H.east |- Y) to node [below] {$\mathbb{B}$} (u.west |- Y);
        \draw [->-] (u.east |- Y) to (Y);   
        \draw [->-] (u.east |- Y2) [out=0, in=180] to (Y1);   
        
        \node (m) at (-1, 0.7) {}; 
        \node (m1) at (-0.9, 0.7) {}; 
        \node (m2) at (-1.1, 0.7) {}; 
        \draw [->-] (X1) to (m1) ;   
        \node [circle, scale=0.3, fill=black, draw] (o) at (m) {};
        \draw [->-] (m2) to [out=0, in=180] (G.west |-X2) ;        
        \draw [->-] (m2) to [out=0, in=180] (0, 1) to (2.9, 1) to [out=0, in=180] (u.west |-X2) ;       

        \node [rectangle, minimum height=2.5cm, minimum width=6.8cm, draw, dashed, label=$\PP_i$] (BG) at (1.94, 0) {};
    \end{tikzpicture} 
\end{center}

where $u$ is the function that calculates and assigns payoffs and $\bullet$ is the copy operation. There is some magic happening in $u$ here: How do the players reason about their payoff if the payoff function sends no information back to them? The diagram above is actually not the proper string diagram of the full game $\PP_i$, but rather the wiring diagram of the implementation in the Open Game engine, where there is a function \texttt{addRolePayoffs} that can assign a payoff to a player. This allows the game $\G_i$ to be implemented as follows:

\begin{lstlisting}{Haskell}
bigPlayerUnit player2Name payoffBP = [opengame|

   inputs    : inputBigPlayer, inputPiece  ;
   feedback  : ;

   :----------------------------:

   //The BigPlayer offers a slice to Player 2
   inputs    : inputBigPlayer, inputPiece;
   feedback  :  ;
   operation : offerNewSlice_dependent;
   outputs   : offerP1   ;
   returns   :  ;

   //Player 2 responds
   inputs    : player2Name, offerP1   ;
   feedback  : ;
   operation : respondToOffer_dependent ;
   outputs   : responseP2   ;
   returns   :  ;

   //Find the smallest player and give them their payoff
   inputs    : inputBigPlayer, player2Name, offerP1, responseP2;
   feedback  :  ;
   operation : forwardFunction $ orderBySize ;
   outputs   : (bigPlayerName, biggestPiece), (smallPlayerName, smallestPiece);
   returns   :  ;

   inputs    : smallPlayerName, smallestPiece ;
   feedback  :  ;
   operation : addRolePayoffs;
   outputs   :  ;
   returns   :  ;

   //The BigPlayer only gets a payoff if they are the last player
   inputs    : bigPlayerName, biggestPiece * payoffBP ;
   feedback  :  ;
   operation : addRolePayoffs;
   outputs   :  ;
   returns   :  ;

   :----------------------------:

   outputs   : bigPlayerName, biggestPiece ;
   returns   :  ;
   |]
\end{lstlisting}

The full implementation details can be found in \cite{Jansma2022}. The 3-player game is then simply implemented as the sequential composition:

\begin{lstlisting}{Haskell}
bigPlayers_composed_3Players = [opengame|

   inputs    : inputBigPlayer, inputPiece ;
   feedback  : ;

   :----------------------------:

   inputs    : inputBigPlayer, inputPiece ;
   feedback  : ;
   operation : bigPlayer_unit "p2" 0 ;
   outputs   : bigPlayer1, piece1 ;
   returns   :  ;

   inputs    : bigPlayer1, piece1;
   feedback  : ;
   operation : bigPlayer_unit "p3" 1 ;
   outputs   : bigPlayer2, newPiece  ;
   returns   :  ;
   
   :----------------------------:

   outputs   : bigPlayer2, newPiece   ;
   returns   :  ;
   |]
\end{lstlisting}

The equilibrium analysis, included in \cite{Jansma2022}, shows that this game indeed has a fair equilibrium distribution. 

\end{multicols}

\end{document}

%% file: preamble.tex
\usepackage[utf8]{inputenc}
\usepackage{geometry}
\usepackage{multicol}
\usepackage{float}

\usepackage[svgnames]{xcolor}

\usepackage{graphicx}
\graphicspath{{figs}}

\usepackage[parfill]{parskip}

\usepackage{tikz}
\usetikzlibrary{decorations.markings}
\tikzset{->-/.style={decoration={markings, mark=at position .5 with {\arrow{>}}}, postaction={decorate}}}
\usetikzlibrary{shapes.geometric}

\usepackage{amsmath}
\usepackage{amssymb}
\usepackage{amsthm}
\usepackage{mathtools}
\usepackage{dsfont}

\usepackage[svgnames]{xcolor}
\definecolor{Crimson}{rgb}{0.6471, 0.1098, 0.1882}
\definecolor{Aqua}{rgb}{0.098, 0.54117, 0.68627}
\definecolor{cbGreen}{HTML}{4DAC26}
\definecolor{cbRed}{HTML}{D01C8B}

\usepackage[hyphens]{url}
\usepackage{hyperref}
\usepackage[hyphenbreaks]{breakurl}
\hypersetup{
    colorlinks,
    citecolor=Crimson,
    filecolor=black,
    linkcolor=Aqua,
    urlcolor=Aqua
}

\usepackage{caption}
\usepackage{subcaption}

\usepackage{courier}
\usepackage{listings}

\DeclareCaptionFont{white}{\color{white}}
\DeclareCaptionFormat{listing}{\colorbox{gray}{\parbox{\textwidth}{#1#2#3}}}
\lstdefinestyle{mystyle}{
    basicstyle=\ttfamily\scriptsize,
    breakatwhitespace=false,         
    breaklines=true,                 
    captionpos=t,                    
    keepspaces=true,                 
    showspaces=false,                
    showstringspaces=false,
    showtabs=false,                  
    tabsize=2
}

\usepackage{quiver}

\usepackage{booktabs}

\def\ie{{\it i.e.}\ }
\def\eg{{\it e.g.}\ }

\newcommand\G{\mathcal G} 
\newcommand\PP{\mathcal P}

\newtheorem{definition}{Definition}
\newtheorem{theorem}{Theorem}